\documentclass[a4paper]{amsart}
\date{November 11, 2013}
\dedicatory{Dedicated to Leonid Pastur on the occasion of his 75th birthday}
\author{Verena von Conta}
\author{Heinz Siedentop}
\address{Mathematisches Institut\\
 Ludwig-Maximilians-Universit\"at M\"unchen\\
 Theresienstra\ss e 39\\ 80333 M\"unchen\\ Germany}
\title[Atoms in Momentum Space]{Statistical Theory of the Atom in Momentum Space}

\newcommand{\ri}{\text{i}} 
\newcommand{\rd}{\mathrm{d}} 
\newcommand{\rz}{\mathbb{R}}
\newcommand{\cz}{\mathbb{C}}

\def\cA{\mathcal{A}}
\newcommand{\cE}{\mathcal{E}}
\def\cI{\mathcal{I}}
\def\cJ{\mathcal{J}}
\def\cK{\mathcal{K}}
\def\cR{\mathcal{R}}

\def\const{\mathrm{const}}
\def\gtf{\gamma_{\mathrm{TF}}}
\def\tt{\tilde\tau}
\def\tr{\mathrm{tr}}

\def\tf#1{\mathcal{E}_{\mathrm{TF}}({#1})}
\def\itf#1{\mathcal{E}_{\mathrm{mTF}}({#1})}
\def\supp#1{\mathrm{supp}({#1})}

\newtheorem{theorem}{Theorem}
\newtheorem{lemma}{Lemma}

\begin{document}
\maketitle

\begin{abstract}
  We investigate the momentum energy functional for atoms found by
  Englert \cite{Englert1992} who also discussed the relation to the
  Thomas-Fermi functional. We prove that the momentum functional
  yields upon minimization the same value as the well-known
  Thomas-Fermi functional. In fact, we show an explicit relation
  between the minimizers of the functionals. Moreover, it turns out
  that the atomic momentum density converges on the scale $Z^{2/3}$ to
  the minimizer of the momentum energy functional.
\end{abstract}

\section{Introduction}

Atoms are described by the Hamiltonian
\begin{equation}
  \label{eq:atom-hamil}
  H_{N} := \sum_{n=1}^N( -\Delta_n-\frac Z{|x_n|}) + \sum_{1\leq m<n\leq N}{1\over|x_m-x_n|}
\end{equation}
self-adjointly realized in $\bigwedge_{n=1}^N L^2(\rz^3:\cz^q)$.  Here
$q$ is the number of spin states per electron, i.e., $q=2$ in the
physical case.  Already in the early days of quantum mechanics it was
obvious that an explicit solution for the ground state is impossible
and appropriate methods are needed to find the ground state energy and
its density, at least approximately. One of the most influential
methods is the \textquotedblleft statistical\textquotedblright method
developed by Fermi \cite{Fermi1927,Fermi1928} and Thomas
\cite{Thomas1927}: it is given by the Thomas-Fermi functional (Lenz
\cite{Lenz1932})
\begin{equation}
  \label{eq:tf}
  \tf\rho:= \cK(\rho) - \cA(\rho) + \cR(\rho)=\tfrac35 \gtf\int_{\rz^3}\rho(x)^{5/3}\rd x -
  \int_{\rz^3}\frac Z{|x|} \rho(x)\rd x + D[\rho]
\end{equation}
where $D[\rho]$ is the quadratic form of
\begin{equation}
  D(\rho,\sigma):=
  \tfrac12\int_{\rz^3}\rd x \int_{\rz^3}\rd y {\overline{\rho(x)}\sigma(y) \over |x-y|},
\end{equation}
the electrostatic interaction energy of the charge density $\rho$ with
the charge density $\sigma$.  Instead of the high dimensional problem
given by \eqref{eq:atom-hamil} the functional $\tf\rho$ depends on the
one-particle density, i.e., a quantity in three variables, only. The
guiding idea was that the minimum of the functional
\begin{equation}
  \label{minimum}
  E_\mathrm{TF}(Z) = \inf\{\tf\rho| \rho\in L^{5/3}(\rz^3), D(\rho,\rho)<\infty, \rho\geq0\}
\end{equation}
gives an approximation for the energy and its minimizer is an
approximation for the ground state density. That this is in fact true
is Lieb and Simon's \cite{LiebSimon1973,LiebSimon1977} celebrated
result (see also Lieb \cite{Lieb1981}). In the atomic case their
result for the energy says that the infimum of the Thomas-Fermi energy
and the ground state energy of $H_Z$ are equal up to an error of
$o(Z^{7/3})$, i.e.,
\begin{equation}
  \label{eq:grenz}
  {\inf\sigma(H_Z)- E_\mathrm{TF}(Z)\over Z^{7/3}} \to 0.
\end{equation}
Moreover they show that the rescaled quantum density $ \rho$ converges
to the rescaled Thomas-Fermi minimizer $\rho_{Z}$, i.e.,
  \begin{equation}
    \label{eq:density}
    Z^{-2} \rho(\cdot Z^{-1/3}) \to Z^{-2}\rho_Z(\cdot Z^{-1/3}) = \rho_{1}
  \end{equation}
  weakly in the limit $Z\to\infty$. Such a result is of great value to
  determine the linear response of atoms to perturbations that are
  local in position space.

  However, this is not applicable to perturbations that are momentum
  dependent since the momentum density -- as is well-known -- is not
  merely the Fourier transform of the position density.  To compute
  the linear response to such perturbations the position space
  statistical model of the atom is clearly inadequate. Englert
  \cite{Englert1992} realizing this inadequacy derived an energy
  functional for the ground state energy of an atom with atomic number
  $Z$ depending on the momentum density $\tau$ which remedies this
  problem and allows -- in a natural way -- for the treatment of
  purely momentum dependent potentials.

  It reads for fermions having $q$ spin states each
  \begin{multline}\label{ImpFunk}
    \itf\tau:= \cK_m(\tau) - \cA_m(\tau) + \cR_m(\tau) = \int_{\rz^3}
    \xi^2\tau(\xi) \rd\xi
    -  \tfrac32 \gtf^{-\frac12} Z\int_{\rz^3}\tau(\xi)^{2/3}  \rd\xi \\
    + \tfrac34\gtf^{-\frac12} \int_{\rz^3} \rd \xi \int_{\rz^3} \rd
    \eta \big(\tau_<(\xi,\eta)\tau_>(\xi,\eta)^{2/3}-\tfrac15
    \tau_<(\xi,\eta)^{5/3}\big)
  \end{multline}
  where $\gtf:=(6\pi^2/q)^{2/3}$ is the Thomas-Fermi constant,
  $\tau_<(\xi,\eta):= \min\{\tau(\xi),\tau(\eta)\}$, and
  $\tau_>(\xi,\eta):= \max\{\tau(\xi),\tau(\eta)\}$.

  The aim of this article is to establish the basic mathematical
  properties of $\cE_\mathrm{mTF}$ and to show that it yields not only
  the asymptotically correct energy but in fact does also give the
  asymptotically correct momentum density.

  \section{Domain of Definition of the Momentum Functional and Euler
    Equation}
\begin{theorem}
  The functional $\cE_{\mathrm{mTF}}$ is well-defined on real-valued
  functions in $L^1(\rz^3, (1+\xi^2)\rd\xi)$.
\end{theorem}
\begin{proof}
  The kinetic energy $\cK_m$ is obviously well-defined. The finiteness
  of the attraction $\cA_m$ follows from
  \begin{equation}\label{km}
    \int\rd\xi |\tau(\xi)|^{2/3}
    \leq \left(\int{\rd \xi \over(1+\xi^2)^2}\right)^{1/3} \left(\int\rd \xi
      (1+\xi^2)|\tau(\xi)|\right)^{2/3} <\infty
  \end{equation}
  by H\"older's inequality. The finiteness of the repulsion $\cR_m$
  follows from the finiteness of its first contribution which, in
  turn, follows from \eqref{km}, since
  \begin{equation}
    \iint \rd \xi \rd \eta |\tau_>(\xi,\eta)|^{2/3} |\tau_<(\xi,\eta)| \leq 2 \int \rd \xi |\tau(\xi)|^{2/3} \int \rd \eta|\tau(\eta)|.  
  \end{equation}
\end{proof}

We set\\
$
\cI= \{\rho\in L^1\cap L^{5/3}(\rz^3)|\rho\geq0\}$, $ \cI_N=\{\rho\in \cI| \int \rho \leq N \}$, $\cI_{\partial N}=\{\rho\in \cI| \int \rho = N \}$\\
for densities in position space and\\
$\cJ\! = \{(1+|\cdot|^2)\tau\in L^1(\rz^3)| \tau\geq0\}$, $\cJ_N \!= \{\tau\in \cJ | \int\! \tau \leq N \}$,
${\cJ_{\partial N}= \{\tau\in \cJ| \int \tau = N \}}$\\
for densities in momentum space.

Although, $\cK_m$ and $-\cA_m$ are convex in $\tau$, the third summand
$\cR_m$ of $\cE_{\mathrm{mTF}}$ is not. We circumvent this problem
by substituting $\tau \rightarrow \tt^{3/2}$, i.e.,
\begin{equation}
  \label{eq:tt}
  \cE_s(\tilde\tau):=
  \itf{\tilde\tau^{3/2}}
\end{equation}
with $\tilde\tau\in L^{3/2}(\rz^3,(1+\xi^2)\rd \xi)$ and
$\tilde\tau\geq0$.
\begin{lemma}
  \label{lem:conv}
  The functional $\cE_s$ is strictly convex. 
\end{lemma}
\begin{proof}
  The first summand is obviously strictly convex, the second is
  linear. It remains to show the convexity of the repulsion
  term. Writing $\theta$ for the Heaviside function, we get
  \begin{align*}
    & \int_0^\infty \rd r \Big(\int_{\rz^3} \rd\xi [\tt(\xi)-r^2]_+\Big)^2\\
    &=\iint_{\rz^6}\rd\xi\rd\eta \int_0^\infty \rd r [\tt(\xi)-r^2]
    [\tt(\eta)-r^2] \theta(\tt(\xi)-r^2)\theta(\tt(\eta)-r^2)\\
    &=\iint_{\rz^6}\rd\xi\rd\eta \int_0^{\tt_<^{1/2}(\xi,\eta)} \rd r
    \big(\tt(\xi)\tt(\eta)-\tt(\xi)r^2-\tt(\eta)r^2+r^4\big)\\
    &=\iint_{\rz^6}\rd\xi\rd\eta \big(\tt(\xi)\tt(\eta)\tt_<(\xi,\eta)^{1/2}
    -\tfrac13\tt(\xi)\tt_<(\xi,\eta)^{3/2}\\
    &\qquad\qquad-\tfrac13\tt(\eta)\tt_<(\xi,\eta)^{3/2}
    +\tfrac15\tt_<(\xi,\eta)^{5/2}\big)\\
    &=\iint_{\rz^6}\rd\xi\rd\eta \big(\tt_>(\xi,\eta)\tt_<(\xi,\eta)^{3/2}
    -\tfrac13\tt_>(\xi,\eta)\tt_<(\xi,\eta)^{3/2}\\
    &\qquad\qquad-\tfrac13\tt_<(\xi,\eta)^{5/2}
    +\tfrac15\tt_<(\xi,\eta)^{5/2}\big)\\
    &=\frac23\iint_{\rz^6}\rd\xi\rd\eta \Big(\tt_<(\xi,\eta)^{3/2}
    \tt_>(\xi,\eta)-\tfrac15 \tt_<(\xi,\eta)^{5/2}\Big).
  \end{align*} 
  Since the first line is obviously convex, it shows the wanted
  result. Note that $\tt_<$ and $\tt_>$ are defined analogously to
  $\tau_<$ and $\tau_>$, respectively.
\end{proof}

\begin{lemma}
  \label{l2}
  Every minimizer of $\cE_{\mathrm{mTF}}$ on $\cJ$ is positive.
\end{lemma}
\begin{proof}
  Assume that $\tau$ is a minimizer of $\cE_{\mathrm{mTF}}$ on $\cJ$
  and suppose that the set $N_\tau:=\{\xi\in\rz^3|\tau(\xi)=0\}$ on
  which $\tau$ vanishes, would not be of measure zero. Then pick any
  $\sigma \in L^1(\rz^3,(1+\xi^2)\rd\xi)$ with
  $\tau(\xi)\sigma(\xi)=0$ for all $\xi$ but non-vanishing on
  $N_\tau$. Furthermore, assume $\varepsilon>0$. Then by the integral
  representation of the interaction term (Lemma~\ref{lem:conv}) and
  the substitution $\tt^{3/2}=\tau$ we get
  \begin{align}
    &\cE_{\mathrm{mTF}}(\tau+\varepsilon\sigma)-\cE_{\mathrm{mTF}}(\tau)\\
        &\begin{aligned} =&-\tfrac32 \gtf^{-\frac12} Z\int_{\rz^3}\rd\xi
     \varepsilon^{2/3}\sigma(\xi)^{2/3}\\
      &+ \tfrac32\cdot\tfrac34\gtf^{-\frac12} \int_0^\infty \rd r
      \Big(\int_{\rz^3}\rd\xi [\tau(\xi)^{2/3}-r^2]_+ +
      \int_{\rz^3}\rd\xi
      [\varepsilon^{2/3}\sigma(\xi)^{2/3}-r^2]_+\Big)^2 \\
      & -\tfrac32\cdot\tfrac34\gtf^{-\frac12} \int_0^\infty \rd r
      \Big(\int_{\rz^3}\rd\xi [\tau(\xi)^{2/3}-r^2]_+\Big)^2+O(\varepsilon)
    \end{aligned}\\
    &\begin{aligned} =&-\varepsilon^{2/3}\tfrac32 \gtf^{-\frac12}
      Z\int_{\rz^3}\rd\xi
      \sigma(\xi)^{2/3}+O(\varepsilon)\\
      &+2\cdot \tfrac34\cdot\tfrac32\gtf^{-\frac12} \int_0^\infty \rd
      r \Big(\int_{\rz^3}\rd\xi [\tau(\xi)^{\frac23}-r^2]_+\Big)
      \Big(\int_{\rz^3}\rd\eta
      [\varepsilon^{\frac23}\sigma(\eta)^{\frac23}-r^2]_+\Big)     
    \end{aligned}\\
    &\begin{aligned} \leq&-\varepsilon^{2/3}\tfrac32 \gtf^{-\frac12}
      Z\int_{\rz^3}\rd\xi
      \sigma(\xi)^{2/3} +O(\varepsilon)\\
      &+\tfrac94\gtf^{-\frac12} \Big[\int\limits_0^\infty \rd r
      \Big(\int\limits_{\rz^3}\!\rd\xi
      [\tau(\xi)^{\frac23}-r^2]_+\Big)^2\Big]^{\frac12}
      \Big[\int\limits_0^\infty \rd r \Big(\int\limits_{\rz^3}\!\rd\eta
      [\varepsilon^{\frac23}\sigma(\eta)^{\frac23}-r^2]_+\Big)^2\Big]^{\frac12}
    \end{aligned}\\
    &\begin{aligned} =&-\varepsilon^{2/3}\tfrac32 \gtf^{-\frac12}
      Z\int_{\rz^3}\rd\xi \sigma(\xi)^{2/3}+O(\varepsilon^{5/6}).
     \end{aligned}
  \end{align}
  For sufficiently small $\varepsilon$ this implies
  $$\cE_{\rm{mTF}}(\tau+\varepsilon\sigma)-\cE_{\rm{mTF}}(\tau)<0,$$
  i.e., $\tau$ cannot be a minimizer.
\end{proof}
\begin{lemma}
  The Euler equation which any minimizer of $\cE_{\mathrm{mTF}}$ on
  $\cJ$ fulfills is
  \begin{equation}
    \label{eq:euler}
    \sqrt\gtf|\xi|^2\tau(\xi)^{1/3}-Z
    +\int_{\rz^3}\rd\eta
    \Big(\tfrac32\tau(\eta)^{2/3}\tau_<(\xi,\eta)^{1/3}-\tfrac12
    \tau_<(\xi,\eta)\Big)=0.
  \end{equation}
\end{lemma}
\begin{proof}
  Instead of deriving the Euler equation for $\cE_\mathrm{mTF}$ we use
  $\cE_s$ (see \eqref{eq:tt}). 

  Assume $\tt$ be the minimizer which is strictly positive because of
  Lemma \ref{l2}. Thus we can pick any $\sigma\in L^{3/2}(\rz^3,
  (1+\xi^2) \rd \xi)$ and $|\sigma| \leq \tt$. Then, for
  $\varepsilon\in[-1,1]$, $\tt+\varepsilon \sigma$ is an allowed trial
  function and the function $F(\varepsilon) =
  \cE_s(\tt+\varepsilon\sigma)$ has a minimum at zero. We show that
  $F$ is differentiable at zero. 

  The first two summands are obviously differentiable. Thus, we concentrate on
  \begin{equation}
    \begin{split}
      T(\varepsilon)&:=\varepsilon^{-1}\int_0^\infty\rd r \left[\left(\int_{\rz^3}\rd \xi[\tt(\xi)+\varepsilon\sigma(\xi)-r^2]_+\right)^2- \left(\int_{\rz^3}\rd \xi[\tt(\xi)-r^2]_+\right)^2\right]\\
      &=\int_0^\infty\rd r \int_{\rz^3}\rd \xi \int_{\rz^3}\rd \eta
      {\left([\tt(\xi)+\varepsilon\sigma(\xi)-r^2]_+ -
          [\tt(\xi)-r^2]_+\right)\over\varepsilon}\\
      & \phantom{ =\int_0^\infty\rd r \int_{\rz^3}\rd \xi \int_{\rz^3}\rd \eta}\cdot\left([\tt(\eta)+\varepsilon\sigma(\eta)-r^2]_++[\tt(\eta)-r^2]_+\right)\\
      &=:\int_0^\infty\rd r \int_{\rz^3}\rd \xi \int_{\rz^3} \rd \eta
      I(\varepsilon,r,\xi,\eta).
    \end{split}
\end{equation}
Since $|a_+-b_+|\leq |a-b|$ for real $a$ and $b$,
$$|\sigma(\xi)| \left([\tt(\eta)+|\sigma(\eta)|-r^2]_++[\tt(\eta)-r^2]_+\right)$$
is an integrable majorant of the integrand independent of
$\varepsilon$. To apply dominated convergence, we split the integral
in two parts, namely the part where the pointwise limit of $I$ exists
and the rest: Thus
\begin{multline}
  \lim_{\varepsilon\to0}T(\varepsilon) \\
  = \lim_{\varepsilon\to0}\int_0^\infty\rd r \int_{\tt(\xi)=r^2}\rd
  \xi \int \rd \eta I(\varepsilon,r,\xi,\eta)
  + \lim_{\varepsilon\to0}\int_0^\infty\rd r \int_{\tt(\xi)\neq r^2}\rd \xi \int \rd \eta I(\varepsilon,r,\xi,\eta)\\
  = 2\int_0^\infty\rd r \int\rd \xi \int \rd \eta
  \sigma(\xi) \theta(\tt(\xi)-r^2) [\tt(\eta)-r^2]_+.
\end{multline}
Indeed, this proves that $F$ is differentiable. Integration in 
$r$ yields
\begin{multline}
  \int_{\rz^3}\rd\xi\sigma(\xi)\Big[\tfrac32 \xi^2\tt(\xi)^{1/2}
  -\tfrac32 \gtf^{-\frac12}Z\\
  +\tfrac32\cdot\tfrac34 \gtf^{-\frac12} \int_{\rz^3}\rd\eta
  \big(2\tt(\eta)\tt_<(\xi,\eta)^{1/2}-\tfrac23\tt_<(\xi,\eta)^{3/2}\big)\Big]=0.
\end{multline}
Since $\sigma$ is arbitrary we arrive at the desired Euler equation
\eqref{eq:euler}.
\end{proof}

Since the integrand is nonnegative, the Euler implies the following pointwise bound on any minimizer
\begin{equation}
  \label{eq:schranke}
  \tau(\xi) \leq  \gtf^{-3/2}Z^3|\xi|^{-6}.
\end{equation}

\section{Relation between Position and Momentum Functional} 
In this section we will see that each summand of $\cE_{\mathrm{mTF}}$
is obtained from the corresponding term of $\cE_{\mathrm{TF}}$ by mere
substitution -- at least for spherically symmetric decreasing
densities. To this end we set
\begin{align}
  S : L^1(\rz^3) &\rightarrow L^1(\rz^3)\\
  \tau&\mapsto \rho
\end{align}
where for all $x\in\rz^3$
\begin{equation}
  \label{eq:S}
  \rho(x) :=  \frac q{(2\pi)^3} \int_{|x| < \gtf^{1/2}|\tau(\xi)|^{1/3}}\rd \xi.
\end{equation}
Moreover, given $\rho\in L^1(\rz^3)$ we define its Fermi radius $r$ by
\begin{equation}
      \label{eq:r}
      r(s):= \sup\{|y| \mid \gtf^{1/2}|\rho(y)|^{1/3} \geq s \text{ for  a.e. } y\in\rz^3\}.
    \end{equation}
This allows to define the operator 
\begin{align}
  T:L^1(\rz^3)&\rightarrow L^1(\rz^3)\\
       \rho&\mapsto \tau
\end{align} 
where for all $\xi\in\rz^3$
\begin{equation}
  \label{eq:T}
  \tau(\xi) := \gtf^{-3/2} r(|\xi|)^3.
\end{equation}

Our first result is
\begin{theorem}\label{haupt} \hspace{.1cm}
  \begin{enumerate}
  \item For all $N\geq 0$ we have $\inf\itf{\cJ_{\partial N}} =
    \inf\tf{\cI_{\partial N}}$.
  \item Assume $N\leq Z$ and $\rho_N$ the minimizer of
    $\cE_{\mathrm{TF}}$ on $\cI_{\partial N}$. Then $T(\rho_N)$ is the
    unique minimizer of $\cE_{\mathrm{mTF}}$ on $\cJ_{\partial N}$.
  \item For $N>Z$ there exists no minimizer of $\cE_{\mathrm{mTF}}$ on
    $\cJ_{\partial N}$.
  \item There exists a unique minimizer $\tau_N$ of
    $\cE_{\mathrm{mTF}}$ on $\cJ_N$. Moreover, $\tau_N \in
    \cJ_{\partial\min\{N,Z\}}$.
  \end{enumerate}
\end{theorem}

To prove Theorem~\ref{haupt} we need a few preliminary results on the
transforms $S$ and $T$ and the way they relate the two functionals
$\cE_{\mathrm{TF}}$ and $\cE_{\mathrm{mTF}}$.
\begin{lemma} \label{ST} \hspace{.1cm}
  \begin{enumerate}
   \item The operators $S$ and $T$ are isometric on $L^1$.
  \item All elements in the image of $S$ and $T$ are spherically symmetric,
    nonnegative, and decreasing.
  \item For every spherically symmetric decreasing $\tau\in\cJ$ 
    $$\cE_{\mathrm{mTF}}(\tau) = \cE_{\mathrm{TF}}\circ S (\tau).$$
  \item For every spherically symmetric decreasing $\rho \in
    \cI$
    $$\cE_\mathrm{mTF}\circ T(\rho)=\tf\rho.$$
\end{enumerate}
\end{lemma}
\begin{proof}
  1.~The claim for $S$ follows easily by
  direct computation interchanging the $x$ and $\xi$ integration.
  
  To treat $T$ we may, without loss of generality, assume
  $\rho\geq0$. Moreover,
  \begin{equation}
    |x|<r(|\xi|)  \Rightarrow \gtf^{1/2}\rho(x)^{1/3} \geq |\xi| 
    \label{Hannelore}
\end{equation}
by definition of $r$ and since $\rho$ is monotonically decreasing.
Furthermore, the definition of $r$ provides
\begin{equation}
    |x|\leq r(|\xi|) \Leftarrow \gtf^{1/2}\rho(x)^{1/3}
  \geq |\xi|.\label{Helmuth}
  \end{equation}

We have
  \begin{equation}
    \label{eq:13}
    \|T(\rho)\|_1
    = \tfrac3{4\pi\gtf^{3/2}}\int \rd \xi \int_{|x|<r(|\xi|)} \rd x 
    = \tfrac3{4\pi\gtf^{3/2}} \int \rd x \int_{|x|<r(|\xi|)}\rd \xi.
  \end{equation}
  By \eqref{Hannelore}
  we get the estimate
  \begin{equation}
    \label{eq:14}
    \|T(\rho)\|_1\leq \tfrac3{4\pi\gtf^{3/2}} \int \rd x \int_{\gtf^{1/2}\rho(x)^{1/3} \geq |\xi|}\rd \xi =\int \rd x \rho(x).
  \end{equation}
  On the other hand, if we allow for $\leq$ instead of strict
  inequality on the integration constraints in \eqref{eq:13} we can
  also reverse  the inequality in \eqref{eq:14} using \eqref{Helmuth}.

  2.~The claims are obvious from the definitions.

  3.~We treat each term of the energy functional individually.  We
  start with the potential terms. Both follow by explicit calculation
  which we exhibit here only for the interaction potential since the
  external potential is an easy variant of it.  Given a radius $a>0$
  we set $K_a:=\chi_{\{x\in\rz^3 \mid\, |x|<a\}}$ to be the
  characteristic function of the ball of radius $a$ centered at the
  origin.  We get
  \begin{align}
    &\cR(S(\tau)) =  \left(\tfrac{q}{(2\pi)^3}\right)^2\iint\rd \xi \rd \eta D(K_{\gtf^{1/2}\tau(\xi)^{1/3}},K_{\gtf^{1/2}\tau(\eta)^{1/3}}) \label{r1}\\
    =& \left(\tfrac3{4\pi}\right)^2\gtf^{-1/2}\iint\rd \xi \rd \eta D(K_{\sqrt[3]{\tau_<(\xi,\eta)}},K_{\sqrt[3]{\tau_>(\xi,\eta)}})\label{r2}\\
    =&\tfrac9{(4\pi)^2}\gtf^{-\frac12}\iint D[K_{\sqrt[3]{\tau_<(\xi,\eta)}}]+ D(K_{\sqrt[3]{\tau_<(\xi,\eta)}},K_{\sqrt[3]{\tau_>(\xi,\eta)}}-K_{\sqrt[3]{\tau_<(\xi,\eta)}})\rd \xi\rd \eta\label{r3}\\
    =&\tfrac{9}{(4\pi)^2}\gtf^{-\frac12}\iint
    D[K_1]\tau_<(\xi,\eta)^{\frac53} +
    \tfrac{4\pi}{2\cdot3}\tau_<(\xi,\eta)2\pi
    \big(\tau_>(\xi,\eta)^{\frac23}
    -\tau_<(\xi,\eta)^{\frac23}\big) \rd \xi \rd \eta\label{r4}\\
    =&\tfrac34\gtf^{-1/2}\iint
    \tau_<(\xi,\eta)\tau_>(\xi,\eta)^{\frac23}- \tfrac15
    \tau_<(\xi,\eta)^{\frac53}\rd \xi \rd \eta\label{r5}
\end{align}
where we used the scaling properties of $D$ and Newton's theorem
\cite{Newton1972}.

The kinetic energy transforms as 
\begin{equation}
  \begin{split}
  \cK(S (\tau)) &=\tfrac35 \gtf \int \rd x S(\tau)(x)^{5/3}\\
&=3 \gtf \int \rd x \int_{t\leq S(\tau)(x)^{1/3}} \rd t t^4 
= \tfrac3{4\pi}\gtf \int \rd \xi \xi^2 \int_{|\xi|^3\leq S(\tau)(x)}\rd x.
\end{split}
\end{equation}
Given that $\tfrac3{4\pi}\int_{|x|<\tau(\eta)^{1/3}}\rd \eta \geq |\xi|^3$ implies $ \tau(\xi)^{1/3}\geq |x|$, we have
\begin{equation}\label{eq:24}
   \tfrac35 \gtf \int \rd x S(\tau)(x)^{5/3}\leq \tfrac3{4\pi} \gtf \int \rd \xi \xi^2 \int_{|x|\leq \gtf^{1/2} \tau(\xi)^{1/3}}\rd x = \int \xi^2 \tau(\xi) \rd \xi.
\end{equation}
Suppose $\tfrac3{4\pi}\int_{|x|<\tau(\eta)^{1/3}}\rd \eta \geq |\xi|^3$
would not imply $\tau(\xi)^{1/3}\geq |x|$. Then
\begin{equation}
  |\xi|^3 \leq \tfrac3{4\pi}\int_{|x|<\tau(\eta)^{1/3}}\rd \eta< \tfrac3{4\pi}\int_{\tau(\xi)<\tau(\eta)}\rd \eta \leq \tfrac3{4\pi}\int_{|\xi|>|\eta|}\rd \eta = |\xi|^3
\end{equation}
where we used in the last inequality that $\tau$ is spherically
symmetric and decreasing.

On the other hand, $\tfrac3{4\pi}\int_{|x|<\tau(\eta)^{1/3}}\rd
\eta \geq |\xi|^3$ follows from $\tau(\xi)^{1/3}>|x|$ as
\begin{equation}
  |\xi|^3= \tfrac3{4\pi} \int_{|\eta|\leq |\xi|}\rd \eta \leq \tfrac3{4\pi} 
  \int_{\tau(\xi) \leq \tau(\eta)}\rd \eta \leq \tfrac3{4\pi}
  \int_{|x|<\tau(\eta)^{1/3}}\rd \eta
\end{equation}
using again that $\tau$ is spherically symmetric and decreasing in the
first inequality.  Thus we can reverse the inequality in
\eqref{eq:24}, i.e.,
\begin{equation}
  \tfrac35 \gtf \int \rd x S(\tau)(x)^{5/3}\geq  \tfrac3{4\pi} \gtf 
  \int \rd \xi \xi^2 \int_{|x|< \gtf^{1/2} \tau(\xi)^{1/3}}\rd x 
  = \int \xi^2 \tau(\xi) \rd \xi.
\end{equation}

4. To prove that $\cE_{\mathrm{mTF}}\circ T(\rho) =\tf\rho$ we proceed
as in 3.  We begin with the kinetic energy:
\begin{multline}
  \cK_m(T(\rho)) = \int \xi^2 \gtf^{-3/2}r(|\xi|)^3  \rd\xi
= \tfrac3{4\pi}\gtf^{-3/2} \int\rd\xi \xi^2\int_{|x|<r(|\xi|)}\rd x \\
= \tfrac3{4\pi} \gtf^{-3/2}\int\rd \xi \xi^2\int_{|\xi|\leq\gtf^{1/2}\rho(x)^{1/3}} = \cK(\rho)
\end{multline}
where we used \eqref{Hannelore} and \eqref{Helmuth} in the penultimate
identity.

We skip again $\cA_m$ and go directly to $\cR_m$. Set
$r_<(|\xi|,|\eta|):=\min\{r(|\xi|), r(|\eta|)\}$ and
$r_>(|\xi|,|\eta|):=\max\{r(|\xi|), r(|\eta|)\}$. Then
\begin{equation}
  \cR_m(T(\rho)) = 
  \left(\tfrac3{4\pi}\right)^2 \gtf^{-1/2} \iint \rd \xi \rd \eta D(K_{\gtf^{-1/2}r_<(|\xi|,|\eta|)}, K_{\gtf^{-1/2}r_>(|\xi|,|\eta|)})
\end{equation}
adapting the steps \eqref{r5} to \eqref{r2}. Making the term explicit
and scaling $\gtf^{-1/2}$ out yields
\begin{multline}
  \label{eq:30}
   \cR_m(T(\rho)) = \tfrac12\left(\tfrac3{4\pi}\right)^2 \gtf^{-3} \iint {\rd x \rd y \over|x-y|} \int_{|x|<r(|\xi|)}\rd\xi \int_{|y|<r(|\eta|)}\rd\eta\\
 =\tfrac12 \left(\tfrac3{4\pi}\right)^2 \gtf^{-3} \iint {\rd x \rd y \over|x-y|}
 \int_{|\xi|\leq\gtf^{1/2}\rho(x)^{1/3}}\rd\xi \int_{|\eta|\leq\gtf^{1/2}\rho(y)^{1/3}}\rd\eta = \cR(\rho)
\end{multline}
where we used \eqref{Hannelore} and \eqref{Helmuth} again.

\end{proof}

Note that by the definition of $T$ and the relation between $r$ and
$\tau$ (Equation \eqref{eq:T}), and Theorem \ref{haupt} any bound on
the position space density implies a corresponding bound on the
momentum space density. In particular the bound $\gtf\rho(x)\leq
Z/|x|$ reproduces \eqref{eq:schranke} which we got already from the
Euler equation.

The Sommerfeld bound
\begin{equation}
  \label{eq:sommerfeld}
  \rho(x) \leq {27\over \pi^3\gtf^{3/2}|x|^6}
\end{equation}
implies
\begin{equation}
  \label{eq:impulssommer}
  \tau(\xi) \leq\left({3\over\pi\gtf|\xi|}\right)^{3/2}.
\end{equation}

\section{The Energy under Rearrangement\label{umordnung}} 

The fact that $\tf\rho$ decreases under spherically symmetric
rearrangement of $\rho$ is well-known (Lieb \cite[Theorem
2.12]{Lieb1981}). The same result holds for the momentum functional:
\begin{lemma} \label{m*}
  For $\tau \in \cJ$,
  \begin{equation}
    \itf {\tau^*} \leq \itf \tau
  \end{equation}
where $\tau^*$ is the spherically symmetric rearrangement of $\tau$.
\end{lemma}

\begin{proof}
  The attraction $\cA_m$ is obviously invariant under rearrangement. The
  repulsion $\cR_m$ is -- by definition -- a superposition of rearranged terms
  only, i.e., is also trivially invariant. 

  Now, $\cK_m(\tau) = \int_0^\infty \rd t \int\xi^2 \chi_{\{\xi\in \rz^3 \mid \tau(\xi)>t\}}(\xi)\rd \xi$. 
 Thus it suffices to show that for any $A\subset \rz^3$ with finite measure
$$   \int \xi^2 \chi_A(\xi) \rd \xi \geq \int \xi^2 \chi_{A^*}(\xi) \rd \xi =  \int \xi^2 K_R(\xi)\rd \xi$$
where $R$ is defined by $|A|= (4\pi/3) R^3$, i.e., the radius of the
ball $A^*:=B_R(0)$ centered at the origin which has the same volume as
$A$. Now define $B:= B_R(0)\setminus A$, $C:= A\setminus B_R(0)$, and
$D:=A\cap B_R(0)$. Then $|B|=|C|$, and thus
\begin{align*}
  \int_{A^*} \xi^2 \rd \xi&= \int_B \xi^2 \rd \xi + \int_D \xi^2 \rd
  \xi \leq R^2\int_B \rd \xi+ \int_D \xi^2\rd \xi
  \leq \int_C \xi^2 \rd \xi +  \int_D \xi^2\rd \xi\\
  &= \int_A \xi^2\rd \xi.
\end{align*}
\end{proof}

Now, we can prove Theorem~\ref{haupt}:
\begin{proof}
  1.~Since the energies of both -- momentum and position --
  Thomas-Fermi functionals decrease under spherically symmetric
  rearrangement (Lemma \ref{m*} and \cite[Theorem 2.12]{Lieb1981}) we
  can restrict both functionals to spherically symmetric decreasing
  densities $\rho$ and $\tau$ as far as minimization is
  concerned. Since both $S$ and $T$ preserve the norm, Statement~3 of
  Lemma~\ref{ST} implies that $\inf\itf{\cJ_{\partial N}} \geq
  \inf\tf{\cI_{\partial N}}$ whereas Statement~4 implies the reverse
  inequality. This proves the first assertion of Theorem~\ref{haupt}.

  2.~Since $\cE_{\mathrm{TF}}$ has a unique minimizer $\rho_N$ on
  $\cI_{\partial N}$ (Lieb and Simon \cite[Theorems II.14 and
  II.17]{LiebSimon1977}), it follows from the preceding step 
 and
  Lemma~\ref{ST},4 that $T(\rho_N)$ minimizes $\cE_{\mathrm{mTF}}$ on
  $\cJ_{\partial N}$. It remains to show that there is no other
  minimizer of the momentum functional. This, however, follows from
  strict convexity of $\cE_s$.

  3.~Suppose $\tau_N$ is a minimizer of $\cE_{\mathrm{mTF}}$
  on $\cJ_{\partial N}$ for some $N>Z$. Then $S(\tau_N)$ has to be a
  minimizer of $\cE_{\mathrm{TF}}$ by Statement 1 and Lemma~\ref{ST},3 but
  this does not exist~\cite{LiebSimon1977}.

  4.~Again, if $\tau_N$ minimizes $\cE_{\mathrm{mTF}}$ on $\cJ_N$ then
  $S(\tau_N)$ minimizes $\cE_{\mathrm{TF}}$ on $\cI_N$. Thus, $\int
  \tau_N=\int S(\tau_N)=\min\{Z,N\}$. Uniqueness of $\tau_N$ follows
  from the strict convexity of $\cE_s$.
\end{proof}

\section{Asymptotic Exactness of Englert's Statistical Model of the
  Atom\label{asympt}}

By Theorem \ref{haupt} the infimum of $\itf\tau$ has the same
approximation properties as the atomic Thomas-Fermi functional and --
unlike the Thomas-Fermi functional -- gives the right appropriate
linear response to momentum dependent force. To show the latter, we
define for $\alpha\in \rz$,
\begin{equation}
  \label{eq:hamil}
  H_{N,\alpha} := H_N - \alpha \sum_{n=1}^N \varphi_Z(-\ri\nabla_n)
\end{equation}
with $\varphi_Z(\xi) := Z^{4/3}\varphi(Z^{-2/3}\xi)$. (Later, we will
require some mild conditions on $\varphi$.) Furthermore, we write
$\psi_Z$ for a ground state of the neutral atomic Hamiltonian. (Note
that we assume neutrality largely for convenience.)

We introduce some further useful notations:
\begin{itemize}
\item The one-particle ground state density of any state $\psi$ is
  \begin{equation}
    \rho_\psi(x) :
    = N\sum_{\sigma_1=1}^q...\sum_{\sigma_N=1}^q \int_{\rz^{3(N-1)}}\rd x_2...\rd x_N|\psi(x,\sigma_1;x_2,\sigma_2;\dots;x_N,\sigma_N)|^2
  \end{equation}
  in position space and
  \begin{equation}
    \label{eq:den}
    \tau_\psi(\xi) :
    = N\sum_{\sigma_1=1}^q...\sum_{\sigma_N=1}^q \int_{\rz^{3(N-1)}}\rd \xi_2...\rd \xi_N|\hat\psi(\xi,\sigma_1;\xi_2,\sigma_2;\dots;\xi_N,\sigma_N)|^2
  \end{equation}
  in momentum space.
\item The rescaled density of the ground state is written as
  \begin{equation}
    \label{eq:rescaled}
    \tilde\tau_{\psi_{Z}}(\xi) := Z \tau_{\psi_{Z}}(Z^{2/3}\xi).
  \end{equation}
\item The set of one-particle density matrices is $${\mathcal
  S}:=\{{\gamma\in\mathfrak S}^1 (L^2(\rz^3:\cz^q)) \mid
  0\leq\gamma\leq1\}.$$
\item We write $\gamma_{\psi}$ for the one-particle density matrix of
  an $N$-particle state $|\psi\rangle\langle \psi|$ (with the
  normalization convention $\tr\gamma_{\psi}=N$).
\item We write $e_j$ for the orthonormal eigenvectors of
  $\gamma\in\mathcal S$ and $\lambda_j$ for the eigenvalues. The
  momentum density of $\gamma$ is written as
    $$\tau_\gamma(\xi):=\sum_{\sigma=1}^q\sum_j\lambda_j |\hat
  e_j(\xi,\sigma)|^2.$$
\item We call $\tt_\gamma(\xi):=Z\tau_\gamma(Z^{2/3}\xi)$ the rescaled
  momentum density of $\gamma$.
\item The minimizer of $\cE_{\mathrm{mTF}}$ on $\cJ_Z$ is written as
  $\tau_{Z}$.
\item We call $\phi_Z := Z/|\cdot| - \rho_Z*|\cdot|^{-1}$ the
  Thomas-Fermi potential where $\rho_Z$ is -- as in Theorem \ref{haupt} -- the
  minimizer of the Thomas-Fermi functional.
\item For sufficiently small $\alpha$, we also introduce the effective
  one-particle Hamiltonian
  \begin{equation}
    \label{eq:eff}
    h_{Z,\alpha}:= - \Delta - \phi_Z - \alpha \varphi_Z(\ri^{-1}\nabla).
  \end{equation}
  We write, using the common abuse of notation, $h_{Z,\alpha}(\xi,x)$
  for its symbol (Hamilton function).
\end{itemize}

Our second main result is the limit for the ground state density of
$H_{Z,\alpha}$. This is essential for computing the linear response of
momentum dependent perturbations.
\begin{theorem}
  \label{thm:limit}
  Assume $(1+|\cdot|^{-2}) \varphi\in L^\infty(\rz^3)$ and uniformly
  continuous. Then
  \begin{equation}
    \lim_{Z\to\infty}\int_{\rz^3} \varphi(\xi)\tilde\tau_{\psi_{Z}}(\xi) \rd \xi
    = \int_{\rz^3}  \varphi(\xi)\tau_{1}(\xi) \rd \xi.
  \end{equation}
\end{theorem}
Physically speaking this shows that the momentum density of large
atoms is given asymptotically by Englert's momentum density functional
on the scale $Z^{2/3}$ which is the semiclassical scale with the
effective `Planck' constant $Z^{-1/3}$.

We now turn to the proof of the theorem. First we need a lower bound:
\begin{lemma}
  Assume $0\leq(1+|\cdot|^{-2})\varphi\in L^\infty(\rz^3)$, $\varphi$
  uniformly continuous and $\alpha\in[0,v]$ with $v:=1/(\| |\cdot|^{-2} \varphi\|_\infty)$. Then
  for every $\gamma\in\mathcal S$
  \begin{equation}
    \label{eq:sum}
    \tr(h_{Z,\alpha}\gamma) \geq {Z^{7/3}\over(2\pi)^3} 
    \int\limits_{h_{1,\alpha}(\xi,x)<0} \rd \xi \rd x\ h_{1,\alpha}(\xi,x) - o(Z^{7/3})
  \end{equation}
  uniformly in $\alpha$ for large $Z$.
  
\end{lemma}

\begin{proof}
  First we note that under our hypothesis, both, the allowed phase
  volume ($\nu=0$) and the semiclassical energy ($\nu=1$)
  $$ \int_{h_{Z,\alpha}(\xi,x)<0}\rd \xi \rd x\ h_{Z,\alpha}(\xi,x)^\nu$$
  of the effective Hamiltonian $h_{Z,\alpha}$ are finite. 
  
  Next we follow the lower bound of Lieb's asymptotic result (Lieb
  \cite[Section V.A.2]{Lieb1981}) modified by the additional momentum
  operator $\varphi_Z$. To this end we pick $g\in C^\infty_0(\rz^3)$
  as a spherically symmetric positive function with $\supp g \subset
  K_1$, $\int g^2=1$, and $g_R(x) := R^{3/2}g(Rx)$ its dilatation
  by $R$ which we choose as $R=Z^{1/2}$. Note that
  $\widehat{g_R}=\hat g_{R^-1}$. We have
  \begin{multline}
    \label{ungl}
    \tr(h_{Z,\alpha}\gamma)
    \geq {q\over (2\pi)^3}\int\rd \xi \rd x (h_{Z,\alpha}(\xi,x))_- \\
    - ZR^2 \|\nabla g\|^2 - \tr[(\phi_Z - \phi_Z* |g_R|^2)\gamma]
    - \alpha \tr\{[(\varphi_Z- \varphi_Z*
    |\widehat{g_R}|^2)(-\ri\nabla)]\gamma\}.
    \end{multline}
    The right hand side of the first line is the wanted main term. The
    second line consists of error terms only. The first two error terms
    are of order $O(Z^{7/3-1/30})$ as shown by Lieb. The third is new
    and needs an additional argument. We have
 \begin{multline}
   |\tr[(\varphi_Z(-\ri\nabla) - \varphi_Z*
   |\widehat{g_R}|^2(-\ri\nabla))\gamma] | \\
   \leq Z^{7/3}\int_{\rz^3}\rd \xi\int_{\rz^3}\rd p\, \tt_\gamma(\xi)
   |\varphi(\xi) - \varphi( \xi-p)| |\hat g_{Z^{1/6}}(p)|^2.
 \end{multline}
 However, the integral of the right hand side converges to zero by
 uniform continuity of $\varphi$ and the fact that $\hat g\in
 \mathcal{S}(\rz^3)$. To see this we show that $\|\varphi -
 \varphi*|\hat g_{Z^{1/6}}|^2\|_\infty $ is, for large $Z$,
 arbitrarily small: Pick any positive $\epsilon$, then, there exists a
 $\delta$ such that for all $\xi,p\in \rz^3$ $|p|<\delta$ implies
 $|\varphi(\xi)-\varphi(\xi-p)|<\epsilon$. 

 Note also, that the fact that $\hat g$ is a Schwartz function
 implies, that we have a constant $c>0$ such that for all points
\begin{equation}
  \label{eq:se}
  |\hat g(\xi)|^2 < c/ |\xi|^4
\end{equation}
Moreover, pick $Z$ so large, that $2Z^{-1/6}
c\|\varphi\|_\infty\int_{|p|>\delta}|p|^{-4} < \epsilon/2$.  With this
choice, we estimate
    \begin{multline}
      \int \rd p |\varphi(\xi)-\varphi(\xi-p)| |\hat g_{Z^{1/6}}(p)|^2\\
      = \int_{|p|<\delta} \rd p |\varphi(\xi)-\varphi(\xi-p)| |\hat g_{Z^{1/6}}(p)|^2+ \int_{|p|>\delta} \rd p |\varphi(\xi)-\varphi(\xi-p)| |\hat g_{Z^{1/6}}(p)|^2\\
      \leq \epsilon/2 + 2\|\varphi\|_\infty c
      \int_{|p|>\delta}Z^{1/2}/(Z^{1/6}|p|)^4 \leq \epsilon/2 + \epsilon/2.
    \end{multline}
Thus, for any $\epsilon$ there is a $Z_0$ such that $Z>Z_0$ implies
$$\int_{\rz^3}\rd \xi\int_{\rz^3}\rd p \tt_\gamma(\xi)
|\varphi(\xi) - \varphi( \xi-p)| |\hat g_{Z^{1/6}}(p)|^2 \leq \epsilon
\int\rd \xi \tt_\gamma(\xi)= \epsilon.$$ This shows the claim that the
second line of \eqref{ungl} contains error terms only.
\end{proof}

\begin{proof}[Proof of Theorem \ref{thm:limit}]
  First we remark that it suffices to proof the theorem for positive
  $\varphi$ since we can split $\varphi$ into the part where it is
  strictly positive and strictly negative and do the proof separately
  for those cases.

  \begin{equation}
\begin{split}
  &\alpha Z^{7/3}\int_{\rz^3} \varphi(\xi)\tilde\tau_{\psi_{Z}}(\xi)
  \rd \xi = \langle\psi_Z, H_{Z,0}\psi_Z\rangle - \langle\psi_Z,
  H_{Z,\alpha}\psi_Z\rangle \\
  \leq & \tf{\rho_Z} +\const Z^{11/5}-\left(\langle \psi_Z, \sum_{n=1}^Z
    h_{{Z,\alpha},n}\psi_Z\rangle- D[\rho_Z]- \const \int
    \rho_{\psi_Z}^{4/3}\right)\\
  = &\int \rd \xi \rd x (h_{Z,0}(\xi,x))_-
  - \tr(h_{Z,\alpha}\gamma_{\psi_{Z}}) + \const Z^{11/5}\\
  \leq &\int \rd \xi \rd x (h_{Z,0}(\xi,x))_- -\int \rd \xi \rd x
  (h_{Z,\alpha}(\xi,x))_- + o(Z^{7/3}) \\
  = &\alpha Z^{7/3} \int \rd \xi \varphi(\xi) \tau_1(\xi) -
  \int_{h_{Z,\alpha}(\xi,x)<0\atop h_{Z,0}(\xi,x)>0} \rd \xi \rd x
  h_{Z,\alpha}(\xi,x)+o(Z^{7/3})\\
  \leq &\alpha Z^{7/3} \int \rd \xi \varphi(\xi) \tau_1(\xi) + \alpha
  Z^{7/3}\int_{h_{1,\alpha}(\xi,x)<0\atop h_{1,0}(\xi,x)>0} \rd \xi
  \rd x \varphi(\xi)+o(Z^{7/3}) \\
  = &\alpha Z^{7/3} \int \rd \xi \varphi(\xi) \tau_1(\xi) +
  Z^{7/3}o(\alpha)+o(Z^{7/3})
  \end{split}
\end{equation}
  where we used successively Lieb's asymptotic result on the atomic
  energy and Lieb's correlation inequality\cite{Lieb1979} in the first
  inequality. Next we use \eqref{eq:sum} and in the last two steps we
  dismiss a negative term and estimate the phase integral in the
  energy shell as $o(\alpha)$.

Now, dividing first by $Z^{7/3}$ and tending $Z$  to $\infty$ yields
\begin{equation}
  \alpha \limsup_{Z\to\infty} \int_{\rz^3} \varphi(\xi)\tilde\tau_{\psi_{Z}}(\xi)
  \rd \xi \leq \alpha \int \varphi(\xi) \tau_1(\xi)  \rd \xi+o( \alpha).
\end{equation}
Dividing by $\alpha$ and choosing $\alpha \downarrow 0$ yields the desired upper bound
reversing the sign of $\alpha$ and taking $\alpha\uparrow0$ yields the
reverse inequality for the inferior limit.

\end{proof}

\textit{Acknowledgment}: This work has been partially supported by the
DFG through the SFB-TR 12 ``Symmetries and Universality in Mesoscopic
Systems".

\def\cprime{$'$}

\end{document}